\documentclass[11pt, letterpaper]{article}

\usepackage[margin=1in]{geometry}
\usepackage[left]{lineno}

\usepackage[T1]{fontenc}
\usepackage{lmodern}
\usepackage{amsmath, amssymb, amsthm, dsfont, mathtools}
\usepackage{thmtools}
\usepackage{thm-restate}
\usepackage{xspace}
\usepackage{graphicx}
\usepackage[dvipsnames]{xcolor}
\usepackage{tikz}
\usetikzlibrary{decorations.pathmorphing, arrows.meta}
\usetikzlibrary{decorations.pathreplacing, calc}
\usetikzlibrary{shapes.geometric}
\usetikzlibrary {decorations.markings}
\usepackage{subcaption}
\usepackage{nicefrac}
\usepackage{bm}
\usepackage{thm-restate}
\usepackage{enumitem}
\usepackage{aligned-overset}

\theoremstyle{definition}
\newtheorem{definition}{Definition}

\theoremstyle{plain}
\newtheorem{theorem}[definition]{Theorem}

\newtheorem{lemma}[definition]{Lemma}
\newtheorem{proposition}[definition]{Proposition}

\newtheorem{claim}[definition]{Claim}

\usepackage{hyperref}
\usepackage[nameinlink]{cleveref}

\crefformat{lemma}{#2Lemma~#1#3}
\crefformat{theorem}{#2Theorem~#1#3}
\crefformat{claim}{#2Claim~#1#3}
\crefformat{proposition}{#2Proposition~#1#3}
\crefformat{equation}{#2(#1)#3}
\crefformat{corollary}{#2Corollary~#1#3}
\crefformat{definition}{#2Definition~#1#3}
\crefformat{figure}{#2Figure~#1#3}
\crefformat{observation}{#2Observation~#1#3}
\crefformat{section}{#2Section~#1#3}
\crefformat{question}{#2Question~#1#3}

\definecolor{LinkViolet}{hsb}{0.37,1,0.5}
\hypersetup{
  colorlinks=true,
  linkcolor=Violet,
  citecolor=Violet,
  urlcolor=green!40!black
}
\tikzset{
  vertex/.style={circle, draw, fill=white, minimum size=4mm, inner sep=0pt,
                 inner sep=0pt, font=\small},
  cur/.style={circle, draw, fill=black!18, minimum size=4.5mm,
                 inner sep=0pt, font=\small},
  epath/.style={thick, decorate,
                 decoration={snake, amplitude=0.5mm, segment length=3mm,
                             pre length=2mm, post length=2mm}},
  cost/.style={fill=white, inner sep=1.5pt, font=\small}}

\newcommand{\N}{\mathbb{N}}
\newcommand{\R}{\mathbb{R}}

\newcommand{\bigO}{\mathcal{O}}

\newcommand{\alg}{\ensuremath{\textsc{Alg}}\xspace}
\newcommand{\opt}{\ensuremath{\textsc{Opt}}\xspace}

\newcommand{\ualg}{u_\alg}
\newcommand{\uopt}{u_\opt}
\newcommand{\uex}{u_\mathrm{Ex}}

\renewcommand{\epsilon}{\varepsilon}

\newcommand{\blocking}{\ensuremath{\textsc{Blocking}}\xspace}

\begin{document}
\title{A lower bound of 4 for online graph exploration}
\author{Júlia Baligács\thanks{Email: \href{mailto:jbaligacs@gmail.com}{jbaligacs@gmail.com}}
\vspace{2mm}\\ University of Oxford}
\date{}
\maketitle              % typeset the header of the contribution
\begin{abstract}
	In the online graph exploration problem, a single agent needs to visit every vertex of an initially unknown graph, which is learned over time in an online fashion, and return to its starting position. We prove that the competitive ratio of this problem is at least 4, improving on the previously best known lower bound of 10/3. 
	A key ingredient of our proof is showing that several restrictions can be imposed on the agent's behavior without affecting the competitive ratio. 
	As a byproduct, we also obtain that certain graph properties, such as the triangle inequality or being subcubic, can be assumed without affecting the competitive ratio.

\end{abstract}

\section{Introduction}

We study the \emph{online graph exploration problem} proposed by Kalyanasundaram and Pruhs~\cite{pruhs94}.
In this setting, a single agent has to traverse an undirected, connected graph~$G=(V,E,w)$ with non-negative edge weights $w\colon E \to \mathbb{R}_{\geq 0}$.
We assume that every vertex and every edge has a unique identifier.
Importantly, the graph is initially unknown and is learned by the agent over time in an online fashion.
More precisely, upon visiting a vertex for the first time, the agent learns the identifiers of the adjacent vertices as well as the identifiers and weights of the corresponding edges.
The \emph{cost} incurred when traversing an edge is simply its weight. 
The objective is to visit all vertices and return to the starting position, while minimizing the total cost.
In this work, we give an improved lower bound on the performance guarantee of any deterministic algorithm for this problem.

As usual in online optimization, we measure the quality of an algorithm in terms of \emph{competitive analysis}.
Given an algorithm \alg for exploration, we denote by $\alg(G,v)$ the total cost incurred on graph~$G$ with starting vertex~$v$.
By~$\opt(G)$, we denote the offline optimum cost, that is, the length of a shortest~TSP tour.
Note that the cost incurred by an online algorithm may depend on the starting position, whereas the offline optimum cost does not.
For~$\rho\geq 1$, we say that \alg is $\rho$-competitive if there exists a constant~$C\in \R_{\geq 0}$ such that~$\alg(G,v)\leq \rho \cdot \opt(G)+C$ for every graph $G$ and starting vertex~$v$ of~$G$.
If~$C=0$ is possible, we say that \alg is \emph{strictly $\rho$-competitive}.
The \emph{(strict) competitive ratio of \alg} is defined by~$\inf \{\rho\geq 1: \alg \text{ is (strictly) $\rho$-competitive}\}$
and the \emph{(strict) competitive ratio of the problem} is defined by~$\inf\{\rho\geq 1: \text{ there exists a (strictly) $\rho$-competitive algorithm}\}$. 
While in general, the strict competitive ratio may be larger than the competitive ratio, we prove in this work that these values coincide for online graph exploration (\cref{obs:exploration-competitive-ratio}).

Arguably, the most important question in the study of the online graph exploration problem is whether there exists a constant-competitive algorithm. This question was proposed by Kalyanasundaram and Pruhs in 1994~\cite{pruhs94} and is still open.
The best known upper bound on the competitive ratio is $\bigO(\log n)$, where $n$ denotes the number of vertices~\cite{NN}, and, prior to our work, the best known lower bound was~10/3~\cite{birx}.
Constant upper bounds are only known when restricting the problem to certain graph classes, such as graphs excluding a fixed minor (including planar, bounded-genus, and bounded-treewidth graphs)~\cite{baligacs}, or graphs that only use a constant number of different weights~\cite{megow}.

\paragraph*{Our results.}

Our main result is an improvement on the best known lower bound on the competitive ratio of the online graph exploration problem.
Here, a graph is \emph{subcubic} if all of its vertices have degree at most 3.

\begin{restatable}{theorem}{explorationlowerbound}\label{thm:4-lower-bound}
	The competitive ratio of the online graph exploration problem is at least 4, even when restricted to subcubic planar graphs.
\end{restatable}

Prior to our work, the best known lower bound was 10/3~\cite{birx}.
While our construction is closely related to that in~\cite{birx}, we provide a refined version that allows for better estimates on the backtracking cost of an agent, that is, the cost incurred when traversing already explored parts of the graph.
Another key ingredient in proving \cref{thm:4-lower-bound} is showing that several restrictions can be imposed on the agent’s behavior without affecting the competitive ratio.
As a byproduct of the restrictions on the agent's behavior, we obtain that certain assumptions can be made about the underlying graph itself without decreasing the competitive ratio. The following result summarizes the most important of these restrictions.

\begin{proposition}
\label{prop:restrictions-intro}
	In the online graph exploration problem, we can assume (any subset of) the following restrictions without affecting the competitive ratio of the problem.
	\begin{itemize}
		\item Upon visiting a new vertex, the agent only learns the unique identifiers and weights of incident edges (but not the identifiers of neighboring vertices).
		\item The graph is subcubic.
		\item The edge weights fulfill the triangle inequality, i.e., $w(e)\leq d(u,v)$ for every edge~$e=\{u,v\}\in E$. Here, $d(\cdot,\cdot)$ denotes the shortest path distance.
		\item The agent is given the total number of vertices beforehand.
	\end{itemize}
\end{proposition}

We hope these insights will be useful for future research on online graph exploration beyond the lower bound construction presented here.

\paragraph*{Related work.}
Several algorithms have been proposed for the online graph exploration problem.
The arguably most natural approach is the \emph{nearest neighbor algorithm}, where, in each step, the agent greedily selects a closest unexplored vertex.
Rosenkrantz, Stearns and Lewis proved that the competitive ratio of this algorithm is~$\Theta(\log n)$~\cite{NN}. Interestingly, the lower bound of $\Omega(\log n)$ is already achieved on unweighted planar graphs~\cite{NNleader} and on trees~\cite{fritsch}.

Another approach is given by \emph{hierarchical depth-first search (hDFS)}, which was introduced by Megow, Mehlhorn, and Schweitzer~\cite{megow}.
This algorithm is constant-competitive on graphs only using a bounded number of different weights and~$\Theta(\log n)$-competitive on general graphs, where the lower bound of $\Omega(\log n)$ is already achieved on a path.
The algorithm \blocking was introduced in~\cite{megow} and builds on the algorithm Shortcut by Kalyanasundaram and Pruhs~\cite{pruhs94}.
In the latter work, the authors showed that the algorithm has competitive ratio at most~16 on planar graphs. In~\cite{megow}, Megow, Mehlhorn, and Schweitzer extended this by proving that \blocking is constant-competitive on bounded-genus graphs, and in~\cite{baligacs}, this is extended even further to graphs excluding a minor.
In~\cite{eberle}, Eberle et al.~revisited the nearest neighbor approach with learning augmentation.

Miyazaki, Morimoto, and Okabe showed that the competitive ratio of online graph exploration is precisely $(1+\sqrt{3})/2$ on cycles and 2 on unweighted graphs~\cite{miyazaki}.
Other graph classes on which the problem has been studied include tadpole graphs \cite{tadpole}, unicyclic graphs~\cite{fritsch,KobayashiLi/24}, and cactus graphs~\cite{fritsch}.

In terms of lower bounds, the first notable result was given in \cite{miyazaki}, where the authors proved that the competitive ratio of online graph exploration on unweighted graphs is at least~2, which matches the trivial upper bound of 2 achieved by DFS.
Dobrev, Královič, and Markou gave a first weighted construction, proving a lower bound of 2.5~\cite{dobrev}. Birx, Disser, Hopp, and Karousatou built on this construction, proving a lower bound of 10/3~\cite{birx}.
In this work, we further improve the lower bound to 4.

There are numerous other variants of graph exploration studied in the literature, such as exploration of directed graphs \cite{albers,deng,fleischer,foerster}, assuming that the agent has limited memory~\cite{fraigniaud05-memory,reingold} but can use pebbles~\cite{disser-pebbles}, or exploration with a team of agents~\cite{Cosson24,CossonMas24,dereniowski,DisserHackfeldKlimm/19,disserlower}.%\cite{AkkerBuchinFoerster/24,Cosson24,CossonMas24,dereniowski,DisserHackfeldKlimm/19,disserlower}.

\section{Equivalent problems}\label{sec:equivalent-problems}

In this section, we prove that we can impose several restrictions on the agent's behavior that will turn out to be useful for our proof of \cref{thm:4-lower-bound}.
First, we introduce some notation used throughout the paper.
During the course of exploration, we say that a vertex is \emph{explored} if it has been visited by the agent and it is \emph{learned} if one of its neighbors or the vertex itself is explored.
A \emph{boundary edge} is an edge with one explored and one unexplored endpoint.
By convention, we denote boundary edges by $e=(u,v)$, where $u$ is explored and $v$ is unexplored.
Otherwise, we denote edges by $e=\{u,v\}$ as usual for undirected graphs.
For two learned vertices $u$ and $v$, we denote by $d(u,v)$ the length of a shortest path from~$u$ to $v$ whose internal vertices are explored. In other words, this is the smallest upper bound on the distance that the agent is aware of.

\subsection{The adversary model}

It is often useful to view an online optimization problem as a game: one player selects an instance and reveals it piece by piece to the other player, who must make online decisions. In this setting, the latter player is the online algorithm, and the first player is referred to as the \emph{adversary}.
For deterministic online algorithms, it does not matter whether the adversary chooses the entire instance beforehand and reveals it piece by piece, or whether the adversary adapts the unrevealed parts based on the algorithm’s decisions: By definition, if a deterministic algorithm is presented with the same instance $I$ twice, it will make the same decisions in both runs. If there exists another instance~$I'$ that differs only in parts not yet revealed, the algorithm cannot distinguish between $I$ and $I'$, and thus behaves identically on both. Since a\xspace$\rho$-competitive algorithm must be $\rho$-competitive on every instance, the adversary can choose whether the algorithm is operating on $I$ or $I'$ once the distinction becomes relevant.

Finding strategies for the adversary can be seen as a ``dual'' online problem, called the \emph{adversary problem}: Here, we have to construct an instance for the given online problem, while receiving information about the algorithm's decisions over time in an online fashion, without being able to modify the parts of the instance already revealed.
The objective is to maximize the ratio~$\alg(I)/\opt(I)$, where~$I$ is the instance constructed by the adversary.
It is immediate that the optimal value achievable in the adversary problem coincides with the strict competitive ratio of the online problem. 
Importantly, note that this model of an adaptive adversary is only applicable when considering deterministic algorithms.

In the context of online graph exploration, the adversary problem can be stated as follows:
We have to construct a weighted graph $G$ adaptively and the counterplayer is an agent that, in each step, moves to a new unexplored vertex.
Once a vertex $v$ is explored, we have to irrevocably determine the identifiers of its adjacent vertices as well as the identifiers and weights of the corresponding edges.
The objective is to maximize the ratio of the total distance traveled by the agent and the length of a shortest TSP tour.

Throughout this paper, we will often use this adversarial perspective to prove lower bounds or equivalences between online problems.
In particular, it allows us to show the following.

\begin{restatable}{observation}{observationcompetitiveratio}\label{obs:exploration-competitive-ratio}
	Let $\rho\in \R_{\geq 1}$.
	There is a $\rho$-competitive algorithm for online graph exploration if and only if there is a strictly $\rho$-competitive algorithm.
\end{restatable}

\begin{proof}
	By definition, every strictly $\rho$-competitive algorithm is also $\rho$-competitive.
	For the other direction, assume that there is no strictly $\rho$-competitive algorithm. 
	Then there exists an adversarial strategy $\cal A$ that constructs, for every exploration algorithm \alg, a graph~$G$ with starting vertex~$v$, such that $\alg(G,v)\geq \rho\cdot  \opt(G) + \epsilon$ for some $\epsilon>0$.
	Consider the following adversarial strategy: Let the starting vertex~$v$ be adjacent to $M$ vertices $v_1, \dots, v_M$ via edges of weight~0. Treat each $v_i$ as the starting vertex of a separate exploration instance and, for each, apply strategy~$\cal A$.
	Let~$G_1, \dots, G_M$ denote the resulting graphs.
	Note that they are not necessarily equal because the algorithm might choose to follow a different strategy when facing the same situation multiple times.
	However, the adversary is able to build a suitable graph for every possible change of strategy.	
	
	Observe that we can treat these as $M$ instances of exploration because during exploration of~$G_i$, the agent does not gain any information about $G_j$ for $j\neq i$.
	We charge the cost incurred in graph~$G_i$ on the $i$-th problem and let $\alg_i$ denote the strategy followed in graph $G_i$. %Switching between the problems is possible at zero cost and the agent does not gain any information on $G_i$ during the exploration of another graph.		
	We obtain overall a graph~$G$ with
	\begin{equation*}
		\alg(G,v)=\sum\limits_{i=1}^M \alg_i(G_i,v_i)\geq M\epsilon + \rho\cdot \sum\limits_{i=1}^M \opt(G_i)=\rho \cdot \opt(G)+M\epsilon.
	\end{equation*}
	Letting $M\to \infty$, this shows that no $\rho$-competitive algorithm exists.
 \end{proof}

In the remainder of the paper, we therefore use the competitive ratio of the problem and the strict competitive ratio of the problem interchangeably.
In the next subsection, we further make use of the adversary model to prove that some variants of the online graph exploration problem are equivalent.

\subsection{The restricted online graph exploration problem}

We prove that we can impose several restrictions on the agent's behavior without affecting the competitive ratio of the graph exploration problem. This will later be useful for our lower bound construction.
In addition, these restrictions can serve as a ``sanity check'' when developing new algorithms: A sensible algorithm should not rely on any information or capabilities that we restrict in the following.

More precisely, consider the online graph exploration problem, where the agent is additionally restricted, and the adversary is relaxed, as follows.
\begin{enumerate}[label=R\arabic*)]
	\item Upon visiting a new vertex, the agent only learns the unique identifiers and weights of incident edges (but not the identifiers of neighboring vertices).\label{property:exploration-restriction-endpoint}
	\item The graph may have parallel edges.\label{property:exploration-restriction-parallel}
	\item When the agent traverses a boundary edge $e=(u,v)$, a new boundary edge~$e'=(u,v')$ of weight~$w(e)$ may appear incident to $u$.\label{property:exploration-restriction-newedges}
	\item If there are two boundary edges $e=(u,v)$ and $e''=(u',v')$ with $d(u',v)< w(e'')$, the agent must not traverse edge $e''$. \label{property:exploration-restriction-triangle}
	\item If there are two boundary edges $e=(u,v),e'=(u,v')$ incident to the same vertex $u$ with $w(e)=w(e')$, the adversary may decide which of the two edges the agent traverses first.\label{property:exploration-restriction-nochoice}
\end{enumerate}
We call this version the \emph{restricted online graph exploration problem}.
Note in the above that none of these restrictions apply to the agent in the offline optimum, e.g., in the situation R5), the offline agent may traverse boundary edges in another order.

Next, we observe that these adaptations do not make the exploration problem harder in the following sense.
While some of the next result is folklore, we provide a proof for the sake of completeness.
For convenience, we treat any super-constant competitive ratio, i.e., depending on some graph parameters such as the number of vertices, as infinity.

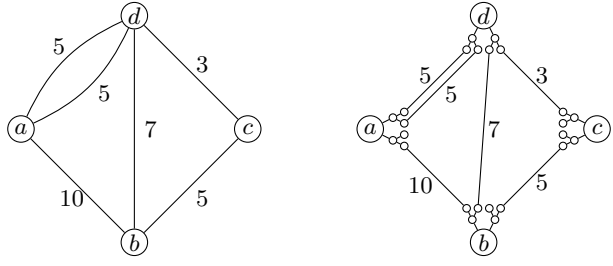
\begin{figure}[t]
	\centering
	\begin{tikzpicture}[scale=0.3]
		\tikzset{every node/.style={font=\footnotesize, inner sep=0pt, minimum size=3.5mm}}
		\node (A) at (0,5) [circle, draw] {$a$};
		\node (B) at (5,0) [circle, draw] {$b$};
		\node (C) at (10,5) [circle, draw] {$c$};
		\node (D) at (5,10) [circle, draw] {$d$};
		\draw (A) to node [left, xshift=3pt, yshift=-5pt] {10}(B) to node [right, xshift=-1pt, yshift=-5pt] {5}(C) to node [right, xshift=-1pt, yshift=3pt] {3} (D) to node [right, xshift=1pt] {7}(B);
		\draw (A) to [bend left=20] node [left, xshift=3pt, yshift=5pt] {5}(D);
		\draw (A) to [bend right=20] node [right, xshift=0pt, yshift=-3pt] {5}(D);
	\end{tikzpicture}
	\hspace{10mm}
	\begin{tikzpicture}[scale=0.3]
		\tikzset{every node/.style={font=\footnotesize, inner sep=0pt, minimum size=1mm, circle, draw}}
		\node (A) at (0,5) [circle, draw, minimum size=3.5mm] {$a$};
		\node (a1) at (1, 4.5) {};
		\node (a2) at (1, 5.5) {};
		\node (a11) at (1.5, 4.25) {};
		\node (a12) at (1.5, 4.75) {};
		\node (a21) at (1.5, 5.25) {};
		\node (a22) at (1.5, 5.75) {};
		\draw (a1) to (A) to (a2);
		\draw (a11) to (a1) to (a12);
		\draw (a21) to (a2) to (a22);
		\node (C) at (10,5) [circle, draw, minimum size=3.5mm] {$c$};
		\node (c1) at (9, 4.5) {};
		\node (c2) at (9, 5.5) {};
		\node (c11) at (8.5, 4.25) {};
		\node (c12) at (8.5, 4.75) {};
		\node (c21) at (8.5, 5.25) {};
		\node (c22) at (8.5, 5.75) {};
		\draw (c1) to (C) to (c2);
		\draw (c11) to (c1) to (c12);
		\draw (c21) to (c2) to (c22);
		
		\node (B) at (5,0) [circle, draw, minimum size=3.5mm] {$b$};
		\node (B1) at (5.5, 1) {};
		\node (B2) at (4.5, 1) {};
		\node (B11) at (5.75, 1.5) {};
		\node (B12) at (5.25, 1.5) {};
		\node (B21) at (4.75, 1.5) {};
		\node (B22) at (4.25, 1.5) {};
		\draw (B1) to (B) to (B2);
		\draw (B11) to (B1) to (B12);
		\draw (B21) to (B2) to (B22);

		\node (D) at (5,10) [circle, draw, minimum size=3.5mm] {$d$};
		\node (D1) at (5.5, 9) {};
		\node (D2) at (4.5, 9) {};
		\node (D11) at (5.75, 8.5) {};
		\node (D12) at (5.25, 8.5) {};
		\node (D21) at (4.75, 8.5) {};
		\node (D22) at (4.25, 8.5) {};
		\draw (D1) to (D) to (D2);
		\draw (D11) to (D1) to (D12);
		\draw (D21) to (D2) to (D22);
		
		\draw (a22) to node [left, draw=none, yshift=2pt] {5} (D22);
		\draw (a21) to node [right, draw=none, yshift=-3pt] {5} (D21);
		\draw (a11) to node [left, draw=none, yshift=-3pt] {10} (B22);
		\draw (B21) to node [right, draw=none, yshift=0pt] {7} (D12);
		\draw (B11) to node [right, draw=none, yshift=-2pt] {5} (c11);
		\draw (c22) to node [right, draw=none, yshift=2pt] {3} (D11);

	\end{tikzpicture}
	\caption{Construction for proving that we can assume R1) and R2). 
		All edges without label have weight 0.
		We have replaced every vertex by a binary tree.
		If, in the left graph, vertex $d$ is unexplored, the agent knows that the boundary edges of weight 3 and 7 lead to the same vertex $d$. 
		In the corresponding situation in the right graph, the agent does not have this information.
	}\label{fig:equivalence-restricted-exploration}
\end{figure}

\begin{restatable}{observation}{rogep}\label{obs:exploration-restrictions}
	The competitive ratio of online graph exploration is equal to the competitive ratio of restricted online graph exploration.
\end{restatable}

\begin{proof}
	It is immediate that a $\rho$-competitive algorithm for the restricted online graph exploration problem is also $\rho$-competitive for online graph exploration.
	For the other direction, assume there exists an adversarial strategy $\cal A$, constructing, for every algorithm for the restricted problem, a multigraph $G$ with starting vertex $v$ such that~$\alg(G,v) > \rho \cdot \opt(G)$. 
	We give an adversarial strategy that achieves the same lower bound of $\rho$ in the classical online graph exploration setting. %; that is, for every non-restricted algorithm \alg, there exists a graph $G$ with starting vertex $v$, such that~$\alg(G,v) > \rho \cdot \opt(G)$.
	For this, we construct a graph that mimics the multigraph constructed by $\cal A$ and argue that, in this graph, any movement of an agent for exploration mimics a movement of some agent for restricted exploration.

	First, we argue that R1) and R2) can be assumed without loss of generality.
	The main idea is to replace each vertex $v$ that $\cal A$ introduces by a suitable graph~$G_v$ called a \emph{gadget}, in which we set all edge weights to 0, and connect all edges incident to~$v$ to different vertices of $G_v$ (see \cref{fig:equivalence-restricted-exploration}). For instance, we can set $G_v$ to be a binary tree with at least $d_v$ leaves, where $d_v$ is the degree of~$v$, and connect the edges to its leaves.
	
	From the perspective of the induced metric space of the resulting graph, all vertices in $G_v$ represent essentially the same vertex $v$ because their pairwise distances are 0.
	In particular, we can assume that, if the agent  explores any vertex of $G_v$, it immediately explores all vertices of $G_v$ and thus, learns about all edges incident to vertices of $G_v$.
	Moreover, this construction does not affect the offline optimum cost.
	
	The key difference is the following: In the resulting graph, the agent cannot tell whether two boundary edges lead to the same gadget $G_v$. In other words, learning the identifiers of the endpoints of boundary edges effectively reduces to learning only the identifiers of the boundary edges themselves.
	Moreover, parallel edges between vertices $u$ and $v$ in the graph constructed by $\cal A$ can be modeled as multiple edges between leaves of $G_u$ and $G_v$.
	This shows that we can assume without loss of generality R1) and R2).
	We argue that the remaining properties can also be assumed without loss of generality using that we have already established this for R1) and R2).
	For the remainder of the proof, we denote boundary edges by  $e=(u, \cdot )$ and, once the agent decides to traverse it, the adversary chooses its endpoint.
	
	The statement that we can assume R5) now immediately follows from R1): both edges have the form $e=(u,\cdot)$ of the same weight, so the agent cannot distinguish between them.
	
	For R3), we use the following construction: for every boundary edge $e=(u,\cdot)$ introduced by $\cal A$, we instead introduce a sufficiently large number of boundary edges of weight $w(e)$, each incident to~$u$, which we call a \emph{bundle}.	
	If the agent traverses $e$, this corresponds to traversing one of the edges of the corresponding bundle.
	Upon the traversal of an edge $e=(u,v)$ in the bundle, the adversary can decide whether all edges of the bundle lead to vertex $v$, or only $e$ leads to vertex $v$ and all other edges of the bundle lead to other vertices with not yet known endpoints, i.e., they remain boundary edges.
	Note that the latter corresponds to the capability of $\cal A$ in R3).

	For property R4), if the agent traverses a boundary edge $e'=(u',\cdot )$ even though a boundary edge~$e=(u,\cdot)$ exists with $d(u',u)+w(e)<w(e')$,
	the adversary lets~$e'$ and $e$ have the same endpoint $v$ and it introduces a new boundary edge incident to~$u'$ of weight $w(e')$ (using R3).
	This corresponds to forcing the agent not to traverse~$e'$: If the agent decides otherwise, the outcome remains the same as if it had chosen to
	traverse $e$, but it incurs a higher cost.
 \end{proof}

The main idea for the proof was to replace each vertex by a binary tree and connect the edges to its leaves (see \cref{fig:equivalence-restricted-exploration}).
If a graph class is closed under this operation, we say that it is \emph{closed under tree replacement}.
This property holds, for example, for the classes of all graphs, planar graphs, bounded-genus graphs, graphs excluding a non-planar minor, and graphs only using a bounded number of different edge weights including~0.
Building on the last observation, we obtain equivalence results for further variants of the problem, which we summarize in the following.

\begin{restatable}{corollary}{corequiv}\label{cor:exploration-restrictions}
	The competitive ratios of the following problems coincide.
	\begin{enumerate}[label=\alph*)]
		\item online graph exploration,
		\item restricted online graph exploration,
		\item restricted and classical online graph exploration on subcubic graphs,
		\item restricted and classical online graph exploration on graphs fulfilling the triangle inequality,
		\item each of the above problems with the relaxation that the agent is given the total number of vertices at the beginning of exploration.
	\end{enumerate}
	The equivalence still holds when restricting all of the problems to the same graph class closed under tree replacement.
\end{restatable}

\begin{proof}
	Equivalence of a) and b) is precisely \cref{obs:exploration-restrictions} and c) follows from the fact that the graphs constructed in the proof of \cref{obs:exploration-restrictions} are subcubic.
	For d), assume that an input graph contains an edge $e=\{u,v\}$ not fulfilling the triangle inequality, i.e., $w(e)>d(u,v)$, where the distance is measured here in the final graph.
	As long as $e$ is a boundary edge, there exists another boundary edge on the shortest path from $u$ to $v$ that is cheaper for the agent to traverse. Due to property \ref{property:exploration-restriction-triangle}, $e$ is never traversed so we can assume that the graph does not contain such edges.
	
	For part e), let $\alg_N$ be the strategy that the agent follows when given the information that the final graph consists of~$N$ vertices and assume that $\alg_N$ is $\rho$-competitive on every graph on~$N$ vertices. Then $\alg_N$ is also $\rho$-competitive on every graph $G$ on $n\leq N$ vertices:
	The adversary can attach to the last explored vertex of $G$ a path of $N-n$ vertices using edges of weight 0 so that the agent technically
	operates on a graph on $N$ vertices, however, its incurred cost on $G$ equals the incurred
	cost in the modified graph.
	This shows that $\alg_N$ is also $\rho$-competitive on every graph on $n\leq N$ vertices.
	For this reason, we can assume that $\alg_{N-1}=\alg_N$ for every~$N\in \N_{\geq 2}$. We obtain that we can assume that the agent's strategy is independent of the number of vertices, even when this information is given beforehand.
	
	The last statement follows immediately from the fact that such a graph class is closed under the constructions in \cref{obs:exploration-restrictions} and under attaching paths.
 \end{proof}

This completes the proof of \cref{prop:restrictions-intro}, which follows from \cref{obs:exploration-restrictions} and \cref{cor:exploration-restrictions}.

We note that the results in this subsection are not applicable when proving super-constant bounds on the competitive ratio. For example, if one aims to prove that an algorithm has competitive ratio $f(n)$ for some function depending on $n$, one loses generality when assuming the above restrictions.
This is because the constructions in the proof increase the number of vertices.

\section{Overview of the proof of \cref{thm:4-lower-bound}}\label{sec:lower-bound-overview}

To prove \cref{thm:4-lower-bound}, we have to do the following:
We give an adversarial strategy that, given an algorithm \alg fulfilling \ref{property:exploration-restriction-endpoint}-\ref{property:exploration-restriction-nochoice} and~$\epsilon>0$, constructs a planar graph~$G$ with starting vertex $v$ such that $\alg(G, v)/\opt(G) \geq 4 - \epsilon$.
This implies that the strict competitive ratio of the restricted exploration problem on planar graphs is at least 4. 
Using 
\cref{obs:exploration-competitive-ratio} and \cref{cor:exploration-restrictions}, this implies the statement of \cref{thm:4-lower-bound}.

Before giving the adversarial strategy in detail, let us introduce the main ideas behind it.
Fix an algorithm \alg for the restricted online graph exploration problem. When we refer to \emph{the agent}, we mean the agent of this online algorithm.
In contrast, the \emph{offline agent} denotes the agent in an offline optimum solution.
The graph $G$ that we construct consists of subgraphs serving as building blocks, which we refer to as \emph{blocks}.
The agent will be forced to traverse almost every block in order to reach every vertex of the graph.
Note that the cost incurred by the agent to traverse a block for the first time might differ from the cost of a second traversal, as the agent gains information during the first traversal that can be used later.
Roughly speaking, we will arrange blocks such that every block is traversed once by the offline agent and twice by the online agent. Moreover, blocks will be built such that the agent pays three times as much as the offline agent for the first traversal.

We want to emphasize that the strategy of a block structure was already used in previous works on lower bounds for online graph exploration~\cite{birx,dobrev}.
Here, we use the same strategy for arranging the constructed blocks in the graph $G$ as in~\cite{birx}.
The blocks that we define are based on the same ideas as in~\cite{birx,dobrev}, but we give a refined structure that allows for a stronger, and arguably simpler, analysis.

\begin{figure}[t]
	\centering
	\begin{tikzpicture}[scale=0.8, every node/.style={circle, draw, inner sep=0pt, minimum size=4.5mm}]
		\node (A) at (0,0) {};
		\node at (0,0.45) [draw=none] {{\footnotesize $\ualg$}};
		\node at (-3,0.45) [draw=none] {{\footnotesize $\uopt$}};
		\node at (4,0.45) [draw=none] {{\footnotesize $\uex$}};
		\node (B) at (1,0) {{\footnotesize $w_1$}};
		\node (C) at (2,0) {};
		\node (D) at (3,0) {};
		\node (E) at (4,0) {{\footnotesize $w_k$}};
		\node (B1) at (-1,0) {{\footnotesize $v_1$}};
		\node (C1) at (-2,0) {};
		\node (D1) at (-3,0) {{\footnotesize $v_i$}};
		\draw (D1) to node [above, draw=none] {1} (C1);
		\draw (C1) to node [above, draw=none] {1} (B1);
		\draw (B1) to node [above, draw=none] {1} (A);
		\draw (A) to node [above, draw=none] {1} (B);
		\draw (B) to node [above, draw=none] {1} (C);
		\draw (C) to node [above, draw=none] {1} (D);
		\draw (D) to node [above, draw=none] {1} (E);
		\node (entry) at (-5.2,0) {$s$};
		\draw[dashed] (E) to [bend left=22] node [below, draw=none] {$k$} (D1);
		
		\node (T1) at (6, 0.5) {{\footnotesize $t_1$}};
		\node (T2) at (6, -0.5)  {{\footnotesize $t_2$}};
		\draw[thick, red] (E) to node [above, draw=none] {$k$} (T1);
		\draw[very thick, red] (E) to node [below, draw=none] {$k$} (T2);
		
		\draw[thick, blue] (D1) to node [above, draw=none] {$k$} (entry);
		\draw[thick, blue] (entry) to [bend left=45] node [above, draw=none] {$k$} (A);	
		\draw[very thick] (-3.5,-1.3) to (4.5,-1.3) to (4.5,0.9) to (-3.5,0.9) to (-3.5,-1.3);
		
		\draw[decorate,decoration={brace,amplitude=6pt, mirror}] (-3,-1.5) to  node[midway, below,yshift=-6pt, draw=none]{$i$} (-0.10,-1.5);
		\draw[decorate,decoration={brace,amplitude=6pt, mirror}] (0.01,-1.5) to  node[midway, below,yshift=-6pt, draw=none]{$k$} (4,-1.5);
	\end{tikzpicture}
	\caption{An adversarially constructed instance for the block traversal problem.
		The blue edges are entry edges and the red edges are exit edges. The dashed edge is only present if the agent does not travel back to $s$.
		The figure shows a block constructed by~$\mathcal{A}_1$ (defined in \cref{lem:recursive block}) when identifying vertices at distance 0 with each other and omitting parallel edges. }\label{fig:block-level-1}
\end{figure}
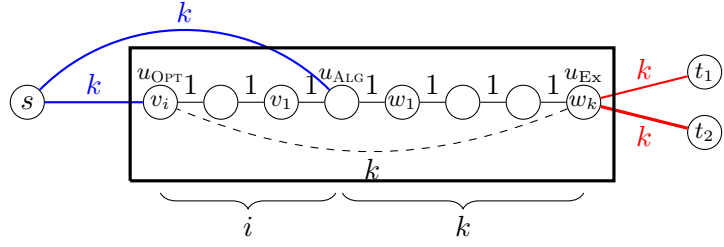

We informally describe the idea for the construction of a (single) block: Consider \cref{fig:block-level-1}.
The agent starts at vertex $s$ and we estimate its incurred cost until it traverses a red edge (without requiring it to visit all vertices).
By \ref{property:exploration-restriction-nochoice}, we can assume that the agent starts by exploring $\ualg$, where it finds itself somewhere along a path of edges of weight 1.
It then explores some of this path until the adversary decides to stop the process and we let $w_k$ be the last explored vertex of the path.
The other endpoint is then $v_i$, which is not yet explored, but~$v_{i-1}$ is.
The value of $i$ depends on the algorithm's behavior.
Note that $i\neq 1$ only occurs if the agent explores the path in ``zigzag movements'', that is, if the agent changes direction before reaching an endpoint of the path.
By \ref{property:exploration-restriction-nochoice}, we can assume that the agent has to explore $v_i$ before it can traverse a red edge. 
Therefore, the agent has to travel along the edge of weight~$k$ to $v_i$ and back before traversing a red edge.
In contrast, the offline agent can visit all vertices by following the path $(s, v_i, \dots, v_1, \ualg, w_1, \dots, w_k, t_1$).

When focusing on the red edges only, we observe the following behavior: Once the agent is located at $w_k$, it incurs a cost of $3k$ to traverse a red edge, whereas the offline agent only incurs a cost of $k$ because it has explored $v_i$ before~$w_k$.
Intuitively speaking, this means that the agent pays three times as much for traversing a red edge as the offline optimum.
Next, we recursively replace all edges of weight 1 in this graph by blocks of the same structure (with edge weights scaled down).
This increases the fraction of edges for which the agent pays three times as much as the offline agent.
In the limit, the agent incurs three times the cost of the offline agent for traversing the block.

Importantly, note that, during a second traversal, the agent does not have to use vertices $v_1, \dots, v_i$, which makes the second traversal shorter. However, the cost incurred during the first traversal increases with $i$, i.e., if the agent explores the path in ``zigzag movements''.
Therefore, in the analysis, we will carefully track these values.
This is one of the key improvements over the work in \cite{birx}.
In \cref{sec:block-traversal-problem}, we formally establish the notion of a block, the problem of traversing it, and the key values that need to be tracked. In \cref{sec:block-construction}, we then construct the blocks described above and prove that the agent incurs three times the cost of the offline optimum when traversing them.

Once we have established blocks, we arrange them in a graph as illustrated in \cref{fig:lower-bound-4}, which is already used in~\cite{birx}.
The structure of this graph is motivated by a lower bound construction for the exploration problem on unweighted graphs, which can be used to prove a lower bound of~2 on this class.
In this graph, the agent travels on some walk from $v_\mathrm{s}$ to $v_1$ (\cref{fig:lower-bound-4}), where it encounters three paths of blocks. Since it cannot distinguish between these paths, we may assume that the first path that is completely explored is the one that leads back to $v_\mathrm{s}$.
The agent then has to backtrack to $v_1$ to make further progress. 
This behavior occurs in every cycle of the graph. 
The agent completes visiting all vertices at some position in the last cycle $C^*$. After this, it has to return to the starting position, incurring additional cost. In total, we obtain that it traverses (almost) every block twice.
By contrast, the offline agent traverses the upper half of each cycle first and then, starting from the last cycle, traverses the lower half of each cycle.
It thus traverses every block precisely once.
This construction is formalized in \cref{sec:lower-bound-exploration-final}.

\section{The block traversal problem}\label{sec:block-traversal-problem}

In this section, we formalize the problem of traversing a block, for which we introduce a new online problem, called the \emph{block traversal problem}. Intuitively, it is a variant of online graph exploration, where the task is only to find the exit of a maze instead of exploring all of it.

An instance of the \emph{block traversal problem} (see \cref{fig:block-level-1} for an example) consists of a weighted connected graph~$B$, called a \emph{block}, together with a starting vertex $s$ and two target vertices $t_1, t_2$ (where~$s$, $t_1, t_2$ are not considered to be part of~$B$).
The vertex~$s$ is connected to two different vertices of $B$, called its \emph{algorithmic entry vertex}~$\ualg$ and its \emph{optimal entry vertex}~$\uopt$, via two edges of the same weight and we call these edges the \emph{algorithmic entry edge} and \emph{optimal entry edge}.
The target vertices $t_1$ and~$t_2$ are connected to the same vertex of $B$, called its \emph{exit vertex}~$\uex$ and we call the corresponding edges the \emph{exit edges}.
By slight abuse of notation, by~$B$ we sometimes refer to the instance of the block traversal problem, instead of only the graph.
In the block traversal problem, an exploration agent obeying rules \ref{property:exploration-restriction-endpoint}-\ref{property:exploration-restriction-nochoice} is initially located in $s$ and is tasked to reach~$t_1$ or~$t_2$ while minimizing the total traveled distance.

Later, a block will be part of our graph in the lower bound construction and we will ensure that the agent
needs to traverse almost all blocks in order to reach every vertex.
Note that with the relaxed requirement that blocks only need to be traversed, instead of the requirement that all vertices are explored, we obtain a safe lower bound on the cost incurred by the agent.
However, for the offline optimum, we still require it to visit all vertices to obtain an upper bound on the cost incurred by the offline agent.

As already hinted at in the previous section, we keep track of multiple values for an instance~$B$ of the block traversal problem that turn out to be useful later on.
In our constructions, we will often connect blocks where the exit edges of one block $B$ correspond to the entry edges of the subsequent block $B'$.
We will charge the cost of traversing these on the block $B'$.
For this reason, when studying $B$ in isolation, we never charge the cost of traversing an exit edge in the following values.

\newcommand{\core}{\ensuremath{\textsc{AlgC}}\xspace}

\begin{itemize}
	\item The \emph{algorithm's cost} $\alg_B$ is the cost incurred by the agent for the block traversal problem on $B$, excluding the cost of traversing an exit edge. %More precisely, \alg produces a walk $W$ from $s$ to $t_1$ or $t_2$, where the second to last vertex in $W$ is $\uex$ (this vertex might appear multiple times in~$W$). Let $W'$ be the walk obtained when removing the last vertex from $W$. Then $W'$ is a walk from $s$ to $\uex$ and we set $\alg_B$ to be the length of this walk.	
	\item The \emph{block length} is $l_B:=d(s,\uex)$.
	In other words, this is the cost incurred by the online or offline agent in a second traversal.
	\item The \emph{optimal exploration cost} $\opt_B$ is the length of a shortest walk from $s$ to~$\uex$ that visits all vertices in $B$. In other words, this is the cost incurred by the offline agent.
	\item The \emph{zigzag length} is  $z_B:=\opt_B-l_B$. Intuitively, this is the ``discount'' that the offline or online agent obtains for a second traversal.
	The name is inspired by the fact that, in the instances that we construct, this value is larger if the agent prefers to travel in ``zigzag movements''.
	\item The \emph{algorithm's core cost} $\core_B$ is defined by $\core_B:=\alg_B-z_B$.
\end{itemize}
The definition of the last value is motivated as follows: In our final construction, blocks will be traversed twice by the agent and we will be able to show that during the first traversal, the incurred cost is $3\opt_B+z_B$, and during the second traversal, the incurred cost is $\opt_B-z_B$.
In a similar spirit to potential function arguments, the adversary ``saves'' $z_B$ in the first traversal and charges it only in the second traversal, i.e., the core cost is the cost charged during the first traversal.
As explained in \cref{sec:lower-bound-overview}, our goal is to prove the following result.
While this is closely related to \cite[Theorem 3.1]{birx}, the key difference is that we estimate the core cost instead of the total algorithm cost.

\begin{theorem}
	\label{thm:3-blocks}
	There exists an adversarial strategy that, for every $\epsilon>0$ and every block traversal algorithm \alg, constructs an instance $B$ of the block traversal problem with~$\core_B/\opt_B\geq 3-\epsilon$.
\end{theorem}

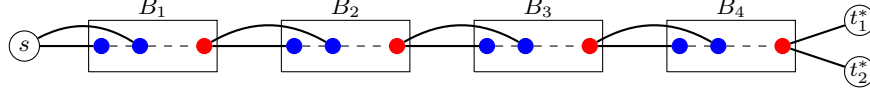
\begin{figure}[t]
	\centering
	\begin{tikzpicture}[scale=1.7, every node/.style={inner sep=0pt, font=\footnotesize}]
		\node (Bex) at (0.5,0) [circle, draw, minimum size=4mm] {$s$};
		\foreach \x in {1, 2.5, 4, 5.5}{
			\node (Bopt) at (\x+0.1,0) [circle, draw, minimum size=2mm, fill, blue] {};
			\node (Balg) at (\x+0.4,0) [circle, draw, minimum size=2mm, fill, blue] {};
			\draw[thick] (Bex) to (Bopt);
			\draw[thick] (Bex) to [bend left] (Balg);
			\node (Bex) at (\x+0.9,0) [circle, draw, minimum size=2mm, fill, red] {};
			\draw[dashed] (Bopt) to (Balg) to (Bex);
			\draw[-] (\x,-0.2) to (\x, 0.2) to (\x+1, 0.2) to (\x+1, -0.2) to (\x, -0.2);
		}
		\node (t1) at (7, 0.2) [circle, draw, minimum size=4mm] {\scriptsize $t_1^*$};
		\node (t2) at (7, -0.2) [circle, draw, minimum size=4mm] {\scriptsize $t_2^*$};
		\draw[thick] (t1) to (Bex) to (t2);
		\node at (1.5, 0.3) {$B_1$};
		\node at (3, 0.3) {$B_2$};
		\node at (4.5, 0.3) {$B_3$};
		\node at (6, 0.3) {$B_4$};
	\end{tikzpicture}
	\caption{A chain of 4 blocks. The construction is obtained by identifying the exit vertex of $B_j$ with the starting vertex of $B_{j+1}$. The algorithmic and optimal entry vertices are depicted in blue, the exit vertices are depicted in red. All other vertices inside the blocks are omitted. The dashed lines indicate that blocks are connected.}\label{fig:chain}
\end{figure}

We will often use that blocks can be connected to each other.
We define a \emph{chain of blocks} (see \cref{fig:chain}) as a sequence of blocks $B_1, \dots, B_i$ where all entry and exit edges have the same weight, by identifying the exit vertex of block $B_j$ with the starting vertex of block~$B_{j+1}$ ($j\in \{1, \dots, i-1\})$.
In particular, we identify the target vertices of block $B_j$ with the algorithmic and optimal entry vertex of block~$B_{j+1}$.
Let~$s$ be the starting vertex of~$B_1$ and $t_1^*, t_2^*$ be the target vertices of the last block $B_i$. Consider the block traversal problem on the chain, which we call in the following~$P$, with $s$ as the starting vertex and  $t_1^*, t_2^*$ as the target vertices.
Note that, while \alg does not learn anything during the traversal of a block~$B_j$ about another block~$B_{j'}$, the traversal of~$B_j$ can still affect the strategy followed in~$B_{j'}$: Indeed, the agent may change strategy when facing the same situation multiple times.
Let $\alg_{B_j}$ denote the strategy followed in block $B_j$.

\begin{restatable}{observation}{observationchain}\label{obs:chains-of-blocks}
	In the block traversal problem on chain $P$, we have
		\begin{enumerate}[label=\alph*)]
				\item $\alg_P\geq \sum_{j=1}^i \alg_{B_j}$,
				\item $\opt_P=\sum_{j=1}^i \opt_{B_j}$,
				\item $l_P=\sum_{j=1}^il_{B_j}$.
			\end{enumerate}
	
\end{restatable}

\begin{proof}
	When the online or offline agent traverses an edge leading from $B_j$ to~$B_{j+1}$, we charge this cost on~$B_{j+1}$.
	The statements now follow from the fact that every path from $s$ to  $t_1^*$ or $t_2^*$ contains the exit vertices of all blocks (where the exit vertex of one block is the starting vertex of the subsequent block).
 \end{proof}

\section{Recursive block construction}\label{sec:block-construction}

We now turn to proving \cref{thm:3-blocks}.
As explained in \cref{sec:lower-bound-overview}, we construct the blocks for this recursively using similar ideas as in \cite{birx}, but estimate the core cost instead of the algorithm's total cost.
One of the key differences is that our blocks have a fixed block length and variable optimum costs, whereas in \cite{birx}, the block length is variable and the optimum cost is fixed.
In our final construction, when the agent has to backtrack, it will have the choice between two different chains containing the same number of blocks.
It will be beneficial for us that these have the same length. % because then we can charge the incurred cost adversarially on any chain.
Another key difference is that we construct only one type of block, where in~\cite{birx}, the beginnings and ends of chains of blocks have a special form so that three types of blocks are needed.
This is not necessary for our construction, which makes it simpler.

Fix a sufficiently large and even integer $k$.
For $d\in\N_0$, we recursively define an adversarial strategy $\mathcal{A}_d$ that constructs an instance $B$ of the block traversal problem with the following properties.

\begin{enumerate}[label=B\arabic*)]
	\item The block length is $l_B=2k^d$.\label{block-property:length}
	\item The optimal cost is bounded by $\opt_B \leq (d+2)k^d$.\label{block-property:opt}
	\item The weight of the entry edges and exit edges is $k^d$.\label{block-property:weights}
\end{enumerate}

We say that $d$ is the \emph{level} of the adversarial strategy, which corresponds to its recursion depth.

\begin{lemma}\label{lem:recursive block}
	For every $d\in \N_0$, there exists an adversarial strategy $\mathcal{A}_d$ of level~$d$ that, for every block traversal algorithm \alg, constructs a planar instance $B$ of the block traversal problem fulfilling B1)-B3) with the following property:
	Let $r_d$ be the infimum of $\core_B/\opt_B$ over all algorithms~\alg, where $B$ is the block constructed by $\mathcal{A}_d$ for \alg.
	Letting $k\in\Omega(d^2)$ and $d\to \infty$, we have~$r_d\to 3$.
\end{lemma}

\begin{proof}
	For $d=0$, we define $\mathcal{A}_0$ to output the block traversal instance illustrated in \cref{fig:block-level-0} for every algorithm \alg.
	It is straightforward that $l_B=\opt_B=\alg_B=2$ so that properties B1)-B3) are fulfilled and we have $r_0=1$.
	Next, we define the strategy $\mathcal{A}_d$ for $d\geq 1$ assuming that  $\mathcal{A}_{d-1}$ is already given.
	
	\begin{figure}[t]
		\centering
		\begin{tikzpicture}[scale=0.75, every node/.style={circle, draw, inner sep=0pt, minimum size=2mm, font=\footnotesize}]
			\node (s) at (-0.5,0) {};
			\node (uopt) at (1,0) {};
			\node (ualg) at (2,0) {};
			\node (uex) at (3,0) {};
			\node (t1) at (4.5, 0.5) {};
			\node (t2) at (4.5, -0.5) {};
			\draw[blue, thick] (s) to node [below, draw=none] {1} (uopt);
			\draw[blue, thick] (s) to [bend left=45] node [above, draw=none, xshift=-5mm] {1} (ualg);
			\draw[red, thick] (uex) to node [above, draw=none] {1} (t1);
			\draw[red, thick] (uex) to node [below, draw=none] {1} (t2);
			\draw[thick] (uopt) to node [above, draw=none] {0} (ualg) to node [above, draw=none] {1} (uex) ;
			\node at (-0.5, -0.3) [draw=none] {$s$};
			\node at (1, -0.3) [draw=none] {$\uopt$};
			\node at (2, -0.3) [draw=none] {$\ualg$};
			\node at (3, -0.3) [draw=none] {$\uex$};
			\node at (4.8, 0.5) [draw=none] {$t_1$};
			\node at (4.8, -0.5) [draw=none] {$t_2$};
			\draw[very thick] (0.5,-0.7) to (3.5, -0.7) to (3.5, 0.7) to (0.5,0.7) to (0.5, -0.7);
		\end{tikzpicture}
		\caption{Block for the adversarial strategy $\mathcal{A}_0$ of level 0.}\label{fig:block-level-0}
	\end{figure}
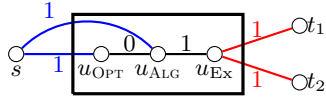
	
	\vspace{2mm}
	\noindent
	\textbf{Construction.}
	The construction is illustrated in \cref{fig:bloc-level-2}.
	The starting vertex~$s$ is incident to the two entry edges of weight~$k^d$ and we may assume that the agent traverses the algorithmic entry edge (using \ref{property:exploration-restriction-nochoice}). 
	Then it is at position~$\ualg$, which we identify with the starting vertices of two blocks in which we use $\mathcal{A}_{d-1}$, i.e., $\ualg$ is incident to four edges of weight~$k^{d-1}$.
	From here, we build two chains of blocks starting from~$\ualg$, i.e., whenever the agent traverses for the first time an exit edge of a block, we present it with a new block, where we identify the exit vertex of the previous block with the starting vertex of the new block.
	In each of these blocks, we apply strategy~$\mathcal{A}_{d-1}$.
	Note that we use here property \ref{block-property:weights} to be able to build these chains.
	In the chain, we say that a block is explored if the agent has traversed its exit edge.
	This procedure stops if either (1) the agent travels back to $s$ or (2) one of the two chains consists of~$(k/2)-1$ explored blocks.
	
	In case (1), let $B_i', \dots, B_1', \ualg, B_1, \dots, B_j$ be the chains of blocks, where $1\leq i, j\leq (k/2)-1$ and all blocks except for $B_i', B_j$ are explored (but $B_i', B_j$ have been entered at some point).
	When the case is triggered, the agent is located at vertex~$s$. 
	Then we complete block $B_i'$ to some valid block constructed by $\mathcal{A}_{d-1}$ and connect its two exit edges both to the same vertex $\uopt$, which we define to be the optimal entry vertex.
	The chain of blocks $ \ualg, B_1, \dots, B_j$ is extended into a chain of~$(k/2)-1$ blocks. 
	Importantly, we only fix here the number of these blocks but do not reveal the structure of the new blocks to the agent. This will be determined when the agent traverses them, where we will again use strategy~$\mathcal{A}_{d-1}$.
	The two exit edges of block $B_{(k/2)-1}$ are then connected to a vertex $w$, and $w$ is connected to a vertex $\uex$ by an edge of weight $k^{d-1}$, which we define to be the exit vertex.
	
	In case (2), let $B_i', \dots, B_1', \ualg, B_1, \dots, B_{(k/2)-1}$ be the chains of blocks, where $i\leq (k/2)-1$ and all blocks but $B_i'$ are explored.
	The case is triggered when the agent traverses an exit edge of block $B_{(k/2)-1}$.
	Then we let both exit edges of $B_{(k/2)-1}$ lead to the same vertex $w$. This is connected by an edge of weight $k^{d-1}$ to the exit vertex $\uex$.
	Block $B_i'$ is completed to some valid block constructed by $\mathcal{A}_{d-1}$ and both of its exit edges lead to the same vertex $\uopt$, which we define to be the optimal entry vertex.
	Additionally, we introduce an edge of weight $k^d$ between~$\uopt$ and~$\uex$.
	The case $d=1$ is illustrated in \cref{fig:block-level-1}.
	
	\begin{figure}[t]
		\centering
		{
			\colorlet{origblue}{blue} % Save original blue
			\colorlet{blue}{green!70!black}
			\colorlet{origred}{red} % Save original blue
			\colorlet{red}{orange!95!black}
			\begin{tikzpicture}[scale=0.495]
				\tikzset{every node/.style={font=\footnotesize}}
				\node (s) at (-0.4,0) [circle, draw, minimum size=2mm, inner sep=0pt] {};
				\node at (-0.5,-0.4) {{\footnotesize $s$}};
				\node (ub) at (1.5,0) [circle, draw, minimum size=2mm, inner sep=0pt] {};
				\node at (1.74,0.7) {{\footnotesize $\uopt$}};
				\node (ua) at (10,0) [circle, draw, minimum size=2mm, inner sep=0pt] {};
				\node at (10,-0.4) {{\footnotesize $\ualg$}};
				\node (uh) at (21,0) [circle, draw, minimum size=2mm, inner sep=0pt]{};
				\node at (21,0.5) {$w$};
				\node (ue) at (22,0) [circle, draw, minimum size=2mm, inner sep=0pt] {};
				\node at (22,0.5) {{\footnotesize $\uex$}};
				
				\node (t1) at (23.5,0.5)  [circle, draw, minimum size=2mm, inner sep=0pt] {};
				\node at (23.5, 1) {$t_1$};
				\node (t2) at (23.5,-0.5)  [circle, draw, minimum size=2mm, inner sep=0pt] {};
				\node at (23.5, -1) {$t_2$};
				
				\node (B1p) at (8.2,0) [rectangle, draw, minimum width=1cm, inner sep=2pt, minimum height=5mm] {{\scriptsize  $B_1'$}};
				\node (B2p) at (5.7,0) [rectangle, draw, minimum width=1cm, inner sep=2pt, minimum height=5mm] {};
				\node (B3p) at (3.2,0) [rectangle, draw, minimum width=1cm, inner sep=2pt, minimum height=5mm] {{\scriptsize  $B_i'$}};
				
				\node (B1) at (11.8,0) [rectangle, draw, minimum width=1cm, inner sep=2pt, minimum height=5mm] {{\scriptsize  $B_1$}};
				\node (B2) at (14.3,0) [rectangle, draw, minimum width=1cm, inner sep=2pt, minimum height=5mm] {{\scriptsize $B_2$}};
				\node (B3) at (16.8,0) [rectangle, draw, minimum width=1cm, inner sep=2pt, minimum height=5mm] {};
				\node (B4) at (19.3,0) [rectangle, draw, minimum width=1cm, inner sep=2pt, minimum height=5mm] {{\scriptsize $B_{\frac{k-2}{2}}$}};
				
				\draw[blue, thick] (ua) to (8.8,0);
				\draw[blue, thick] (ua) to [bend right=25] (8.6,0.25);
				\draw[blue, thick] (7.6, 0) to (6.3,0);
				\draw[blue, thick] (7.6, 0) to [bend right=25] (6.1,0.25);
				\draw[blue, thick] (5.1, 0) to (3.8,0);
				\draw[blue, thick] (5.1, 0) to [bend right=25] (3.6,0.25);
				\draw[blue, thick] (2.6, 0) to (ub);
				\draw[blue, thick] (2.6, 0) to [bend right=40] (ub);
				
				\draw[blue, thick] (ua) to (11.2,0);
				\draw[blue, thick] (ua) to [bend left=25] (11.4,0.25);
				\draw[blue, thick] (12.4,0) to (13.7,0);
				\draw[blue, thick] (12.4,0) to [bend left=25] (13.9,0.25);
				\draw[blue, thick] (14.9,0) to (16.2,0);
				\draw[blue, thick] (14.9,0) to [bend left=25] (16.4,0.25);
				\draw[blue, thick] (17.4,0) to (18.7,0);
				\draw[blue, thick] (17.4,0) to [bend left=25] (18.75,0.3);
				\draw[blue, thick] (19.9,0) to (uh);
				\draw[blue, thick] (19.9,0) to [bend left=40] (uh);
				\draw[blue, thick] (uh) to (ue);
				
				\draw[very thick] (1, -1.2) to (1, 1) to (22.7, 1) to (22.7, -1.2) to (1,-1.2);
				
				\draw[very thick, red] (s) to [in=180, out=60] (1.5, 1.5) to (8, 1.5) to [out=0, in=120] (ua);
				\draw[very thick, red] (s) to (ub);
				\draw[very thick, red, dashed] (ue) to [out=-130, in=0] (20,-0.9) to (3,-0.9) to [out=180, in=-60] (ub);
				\draw[very thick, red] (ue) to (t1);
				\draw[very thick, red] (ue) to (t2);
			\end{tikzpicture}
			\colorlet{blue}{origblue} 
			\colorlet{red}{origred} 
		}
		\caption{Recursive block construction (of level $d$). $B_1, \dots, B_{(k/2)-1}, B_1', \dots, B_i'$ are blocks constructed by $\mathcal{A}_{d-1}$. The thick edges (depicted in orange) have weight $k^d$ and the edges depicted in green have weight $k^{d-1}$.
			The dashed edge between $\uopt$ and $\uex$ is only present in case~(2), i.e., if the agent does not travel back to $s$.
			From this drawing, we obtain that the block is planar.
		}\label{fig:bloc-level-2}
	\end{figure}
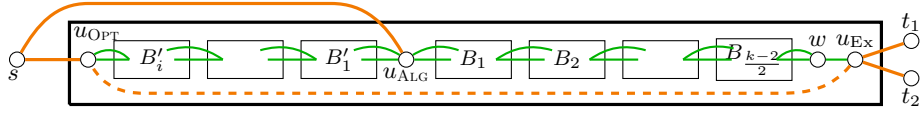
	
	\vspace{1mm}
	\noindent
	\textbf{Block properties and analysis.}
	Note that \ref{block-property:weights} is obviously fulfilled in the construction and one can observe from~\cref{fig:bloc-level-2} that the constructed block is planar, using that, by induction, the blocks constructed by $\mathcal{A}_{d-1}$ are planar as well.
	Using \cref{obs:chains-of-blocks} and that blocks constructed by $\mathcal{A}_{d-1}$ fulfill properties \ref{block-property:length}-\ref{block-property:weights}, we obtain that a shortest path from $s$ to $\uex$ is given by first traversing the algorithmic entry edge, then blocks $B_1, \dots, B_{(k/2)-1}$ and then the edges to~$w$ and $\uex$. We obtain
	\begin{equation}\label{eq:lB-rec}
		l_B=k^d+\sum\limits_{j=1}^{(k/2)-1} l_{B_j}+2k^{d-1}=k^d + (k/2)\cdot 2k^{d-1}=2k^d,
	\end{equation}
	so that property \ref{block-property:length} is fulfilled (note that, in case (2), the existence of the additional edge~$\{\uopt, \uex\}$ does not lead to a shorter path).
	In both cases (1) and (2), observe that
	\begin{equation}\label{eq:opt-level-2}
		\opt_B \leq k^d + \sum\limits_{j=1}^i \opt_{B_j'} + k^{d-1} + \sum\limits_{j=1}^{(k/2)-1} \opt_{B_j} + 2k^{d-1}.
	\end{equation}
	Using that $\opt_{B_j}\leq (d+1)k^{d-1}$, we obtain
	\begin{align*}
		\opt_B&\leq k^d + \left(i+\frac{k}{2}-1\right)\cdot (d+1)k^{d-1}+3k^{d-1}\\
		\overset{3\leq 2(d+1)}&{\leq} k^d + \left(i+\frac{k}{2}+1\right)\cdot (d+1)k^{d-1}
		\overset{i\leq  \frac{k}{2}-1}{\leq} k^d+ (d+1) k^d= (d+2)k^d,
	\end{align*}
	so that property \ref{block-property:opt} is fulfilled as well.
	By definition, we have $z_B=\opt_B-l_B$ and plugging in our findings from \eqref{eq:lB-rec} and \eqref{eq:opt-level-2}, we obtain
	\begin{align}
		z_B&\leq
		\sum\limits_{j=1}^{(k/2)-1} (\opt_{B_j}-l_{B_j})+\sum\limits_{j=1}^{i} \opt_{B_j'} + k^{d-1}
		=\sum\limits_{j=1}^{(k/2)-1} z_{B_j} + \sum\limits_{j=1}^{i} (z_{B_j'}+l_{B_j'}) + k^{d-1}.\label{eq:zB-level-2}
	\end{align}

	Next, we analyze the cost incurred by the agent.
	We say that a block is \emph{backtracked} if the agent traverses it --- either from $s$ to $\uex$ or in the other direction --- after the block was explored.
	In case~(1), the agent traverses three times an entry edge of weight $k^d$ and, before the case was triggered, it explores and backtracks the blocks~$B_1', \dots, B_{i-1}'$. To reach $\uex$, \alg also has to traverse the blocks $B_1, \dots, B_{(k/2)-1}$, where we use that the edge $\{\uopt, \uex\}$ is not present in this case.
	Traversing the exit edge of $ B_{(k/2)-1}$ and $\{w, \uex\}$, the agent incurs another cost of $2k^{d-1}$.
	Using our findings for chains (\cref{obs:chains-of-blocks}), we obtain
	\begin{equation}\label{eq:Alg-recursive-cost}
		\alg_B \geq 3k^d + \sum\limits_{j=1}^{(k/2)-1} \alg_{B_j}+\sum\limits_{j=1}^{i-1}(\alg_{B_j'}+l_{B_j'})+ 2k^{d-1}.
	\end{equation}
	In case (2), the agent traverses one entry edge, explores and backtracks the blocks~$B_1', \dots, B_{i-1}'$ and explores blocks $B_1, \dots, B_{(k/2)-1}$ before the case was triggered. 
	Then, the agent is located at~$w$ and, by property \ref{property:exploration-restriction-nochoice}, the agent is not able to traverse an exit edge before exploring~$\uopt$.
	Traveling from $w$ to~$\uopt$ and then to $\uex$ costs at least~$2k^d+k^{d-1}$.
	Therefore, we obtain the same bound~\eqref{eq:Alg-recursive-cost} for case (2) as in case (1).
	By \eqref{eq:zB-level-2} and \eqref{eq:Alg-recursive-cost}, we obtain a core cost of
	\begin{align}
		\core_B&=\alg_B-z_B
		\geq 3k^d + \sum\limits_{j=1}^{(k/2)-1} (\alg_{B_j}-z_{B_j})+\sum\limits_{j=1}^{i}(\alg_{B_j'}-z_{B_j'})\nonumber\\
		&\hspace{3cm}- \alg_{B_i'}- \underbrace{l_{B_i'}}_{=2k^{d-1}}+ k^{d-1}\nonumber\\
		&= 3 k^d + \sum\limits_{j=1}^{(k/2)-1} \core_{B_j} + \sum\limits_{j=1}^{i} \core_{B_j'} - \alg_{B_i'} - k^{d-1}\nonumber\\
		&\geq 3 k^d + \sum\limits_{j=1}^{(k/2)-1} \core_{B_j} + \sum\limits_{j=1}^{i} \core_{B_j'} - (3d+4) k^{d-1},\label{eq:algc-level-2}
	\end{align}
	where we have used for the last inequality that $\alg_{B_i'}\leq 3 \opt_{B_i'}\leq 3(d+1) k^{d-1}$ because otherwise, the statement of the lemma is clear.
	
	\vspace{2mm}
	\noindent
	\textbf{Ratio analysis.}
	For $d\geq 1$, we obtain using~\eqref{eq:opt-level-2} and \eqref{eq:algc-level-2} that
	\begin{align*}
		r_d&\geq  \frac{3 k^d + \sum_{j=1}^{(k/2)-1} \core_{B_j} + \sum_{j=1}^{i} \core_{B_j'} - (3d+4) k^{d-1}}{k^d + \sum_{j=1}^{(k/2)-1} \opt_{B_j} + \sum_{j=1}^i \opt_{B_j'} + 3k^{d-1}}\\
		&\geq \frac{3 k^d + r_{d-1} \left(\sum_{j=1}^{(k/2)-1} \opt_{B_j} + \sum_{j=1}^{i} \opt_{B_j'}\right) - (3d+4) k^{d-1}}{k^d + \sum_{j=1}^{(k/2)-1} \opt_{B_j} + \sum_{j=1}^i \opt_{B_j'} + 3k^{d-1}}\\
		&\geq \frac{3k^d + r_{d-1}k (d+1) k^{d-1}- (3d+4) k^{d-1}}{k^d+k(d+1)k^{d-1}+3k^{d-1}}
		= \frac{3 + r_{d-1}(d+1) - (3d+4)/k}{1+(d+1)+3/k},
	\end{align*}
	where we have used for the inequality in the last line that $\opt_{B_j}\leq (d+1) k^{d-1}$ by \ref{block-property:opt}, $i<k/2$, and  the fact that $\frac{x+z}{y+z}\leq \frac{x}{y}$ for $x\geq y\geq 0$ and $z\geq 0$.
	From this recursive expression, it is easy to see that, 
	if $k\in \Omega(d^2)$ and we let $d\to \infty$ (and therefore also $k\to \infty$), we have $r_d\to 3$. (This can be seen from the fact that the terms with $k$ in the denominator can be neglected and $r=3$ is the only solution of $r=(3+r(d+1))/(1+d+1) $.)
 \end{proof}

Choosing $d$ large enough in the Lemma, we obtain~\cref{thm:3-blocks}.

\section{Block arrangement}\label{sec:lower-bound-exploration-final}

Now that we have established the necessary preliminaries on block traversal, we turn to proving our main result (\cref{thm:4-lower-bound}) in this section.
For this, we use the same block arrangement as in~\cite{birx}, but we give a refined analysis using our stronger results on the block traversal problem.

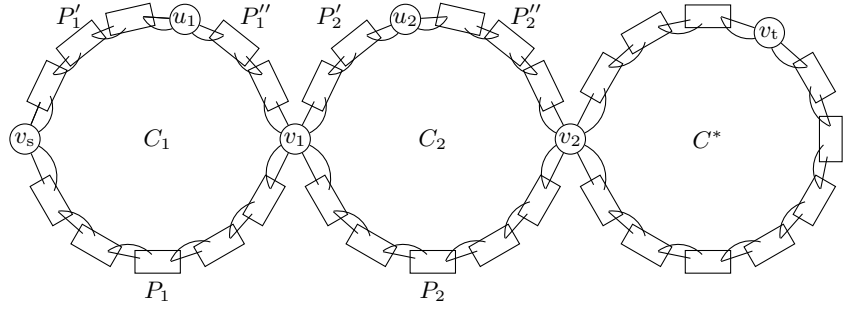
\begin{figure}[t]
	\centering
	\begin{tikzpicture}[scale=0.65]
		\tikzset{every node/.style={font=\footnotesize}}
		\def\centerX{2.5}
		\def\centerY{0}
		\def\radius{2.5} 
		\def\nodes{14}
		\def\rectWidth{6mm}
		\def\rectHeight{3mm}
		
		\foreach \i in {1,2,4,5,6} {
			\pgfmathsetmacro{\angle}{360/\nodes * \i}
			\pgfmathsetmacro{\x}{\centerX + \radius * cos(\angle)}
			\pgfmathsetmacro{\y}{\centerY + \radius * sin(\angle)}
			\pgfmathsetmacro{\tangentAngle}{\angle + 90} 
			\node[draw, minimum width=\rectWidth, minimum height=\rectHeight, rotate=\tangentAngle] at (\x,\y) {};
			\node (R\i) [rectangle, minimum height=1.8mm, minimum width=4.5mm, inner sep=0pt, rotate=\tangentAngle] at (\x,\y) {};
			\pgfmathsetmacro{\dx}{0.25 * cos(\tangentAngle)}
			\pgfmathsetmacro{\dy}{0.25 * sin(\tangentAngle)}
			\node (U\i) [rectangle, minimum height=0.1mm, minimum width=0.1mm, inner sep=0pt, rotate=\tangentAngle] at (\x-\dx,\y-\dy) {};
			\node (S\i) [rectangle, minimum height=0.1mm, minimum width=0.1mm, inner sep=0pt, rotate=\tangentAngle] at (\x+\dx,\y+\dy) {};
		}
		\foreach \i in {7,8,9,10,11} {
			\pgfmathsetmacro{\angle}{360/12 * \i}
			\pgfmathsetmacro{\x}{\centerX + \radius * cos(\angle)}
			\pgfmathsetmacro{\y}{\centerY + \radius * sin(\angle)}
			\pgfmathsetmacro{\tangentAngle}{\angle + 90} 
			\node[draw, minimum width=\rectWidth, minimum height=\rectHeight, rotate=\tangentAngle] at (\x,\y) {};
			\node (R\i) [rectangle, minimum height=1.8mm, minimum width=4.5mm, inner sep=0pt, rotate=\tangentAngle] at (\x,\y) {};
			\pgfmathsetmacro{\dx}{0.25 * cos(\tangentAngle)}
			\pgfmathsetmacro{\dy}{0.25 * sin(\tangentAngle)}
			\node (U\i) [rectangle, minimum height=0.1mm, minimum width=0.1mm, inner sep=0pt, rotate=\tangentAngle] at (\x-\dx,\y-\dy) {};
			\node (S\i) [rectangle, minimum height=0.1mm, minimum width=0.1mm, inner sep=0pt, rotate=\tangentAngle] at (\x+\dx,\y+\dy) {};
		}
		\pgfmathsetmacro{\angle}{360/\nodes * 3}
		\pgfmathsetmacro{\x}{\centerX + \radius * cos(\angle)}
		\pgfmathsetmacro{\y}{\centerY + \radius * sin(\angle)}
		\node (u1) [circle, draw, inner sep=0pt, minimum size=4mm] at (\x, \y) {$u_1$};
		\node (vs) [circle, draw, inner sep=0pt, minimum size=4mm] at (-0.2,0) {$v_\mathrm{s}$};
		\node (v1) [circle, draw, inner sep=0pt, minimum size=4mm] at (5.3,0) {$v_1$};
		\node (v2) [circle, draw, inner sep=0pt, minimum size=4mm] at (10.9,0) {$v_2$};
		
		\draw (vs) to (R6);
		\draw (vs) to [bend right=40] (R6);
		\draw (vs) to (R6);
		\draw (U6) to (R5);
		\draw (U6) to [bend right=40] (R5);
		\draw (U5) to (R4);
		\draw (U5) to [bend right=40] (R4);
		\draw (U4) to (u1);
		\draw (U4) to [bend right=40] (u1);
		\draw (S2) to (u1);
		\draw (S2) to [bend left=40] (u1);
		\draw (S1) to (R2);
		\draw (S1) to [bend left=40] (R2);
		\draw (v1) to (R1);
		\draw (v1) to [bend left=40] (R1);
		
		\draw (vs) to (R7);
		\draw (vs) to [bend left=40] (R7);
		
		\draw (S11) to (v1);
		\draw (S11) to [bend left=40] (v1);
		
		\foreach \x in {7,...,10}{
			\draw (S\x) to (R\the\numexpr \x+1 \relax);
			\draw (S\x) to [bend left=40] (R\the\numexpr \x+1 \relax);
		}
		
		\def\centerX{8.1}
		\def\centerY{0}

		\foreach \i in {1,2,3,5,6} {
			\pgfmathsetmacro{\angle}{360/\nodes * \i}
			\pgfmathsetmacro{\x}{\centerX + \radius * cos(\angle)}
			\pgfmathsetmacro{\y}{\centerY + \radius * sin(\angle)}
			\pgfmathsetmacro{\tangentAngle}{\angle + 90} 
			\node[draw, minimum width=\rectWidth, minimum height=\rectHeight, rotate=\tangentAngle] at (\x,\y) {};
			\node (R\i) [rectangle, minimum height=1.8mm, minimum width=4.5mm, inner sep=0pt, rotate=\tangentAngle] at (\x,\y) {};
			\pgfmathsetmacro{\dx}{0.25 * cos(\tangentAngle)}
			\pgfmathsetmacro{\dy}{0.25 * sin(\tangentAngle)}
			\node (U\i) [rectangle, minimum height=0.1mm, minimum width=0.1mm, inner sep=0pt, rotate=\tangentAngle] at (\x-\dx,\y-\dy) {};
			\node (S\i) [rectangle, minimum height=0.1mm, minimum width=0.1mm, inner sep=0pt, rotate=\tangentAngle] at (\x+\dx,\y+\dy) {};
		}
		\pgfmathsetmacro{\angle}{360/\nodes * 4}
		\pgfmathsetmacro{\x}{\centerX + \radius * cos(\angle)}
		\pgfmathsetmacro{\y}{\centerY + \radius * sin(\angle)}
		\node (u2) [circle, draw, inner sep=0pt, minimum size=4mm] at (\x, \y) {$u_2$};
		
		\draw (v1) to (R6);
		\draw (v1) to [bend right=40] (R6);
		\draw (U6) to (R5);
		\draw (U6) to [bend right=40] (R5);
		\draw (U5) to (u2);
		\draw (U5) to [bend right=40] (u2);
		\draw (U4) to (u1);
		
		\draw (S3) to (u2);
		\draw (S3) to [bend left=40] (u2);
		\draw (S2) to (R3);
		\draw (S2) to [bend left=40] (R3);
		\draw (S1) to (R2);
		\draw (S1) to [bend left=40] (R2);
		\draw (v2) to (R1);
		\draw (v2) to [bend left=40] (R1);
		
		\foreach \i in {7,8,9,10,11} {
			\pgfmathsetmacro{\angle}{360/12 * \i}
			\pgfmathsetmacro{\x}{\centerX + \radius * cos(\angle)}
			\pgfmathsetmacro{\y}{\centerY + \radius * sin(\angle)}
			\pgfmathsetmacro{\tangentAngle}{\angle + 90} 
			\node[draw, minimum width=\rectWidth, minimum height=\rectHeight, rotate=\tangentAngle] at (\x,\y) {};
			\node (R\i) [rectangle, minimum height=1.8mm, minimum width=4.5mm, inner sep=0pt, rotate=\tangentAngle] at (\x,\y) {};
			\pgfmathsetmacro{\dx}{0.25 * cos(\tangentAngle)}
			\pgfmathsetmacro{\dy}{0.25 * sin(\tangentAngle)}
			\node (U\i) [rectangle, minimum height=0.1mm, minimum width=0.1mm, inner sep=0pt, rotate=\tangentAngle] at (\x-\dx,\y-\dy) {};
			\node (S\i) [rectangle, minimum height=0.1mm, minimum width=0.1mm, inner sep=0pt, rotate=\tangentAngle] at (\x+\dx,\y+\dy) {};
		}
		
		\draw (v1) to (R7);
		\draw (v1) to [bend left=40] (R7);
		
		\draw (S11) to (v2);
		\draw (S11) to [bend left=40] (v2);
		
		\foreach \x in {7,...,10}{
			\draw (S\x) to (R\the\numexpr \x+1 \relax);
			\draw (S\x) to [bend left=40] (R\the\numexpr \x+1 \relax);
		}
		
		\def\centerX{13.7}
		\def\centerY{0}

		\foreach \i in {0,1,3,4,5,7,8,9,10,11} {
			\pgfmathsetmacro{\angle}{360/12 * \i}
			\pgfmathsetmacro{\x}{\centerX + \radius * cos(\angle)}
			\pgfmathsetmacro{\y}{\centerY + \radius * sin(\angle)}
			\pgfmathsetmacro{\tangentAngle}{\angle + 90}
			\node[draw, minimum width=\rectWidth, minimum height=\rectHeight, rotate=\tangentAngle] at (\x,\y) {};
			\node (R\i) [rectangle, minimum height=1.8mm, minimum width=4.5mm, inner sep=0pt, rotate=\tangentAngle] at (\x,\y) {};
			\pgfmathsetmacro{\dx}{0.25 * cos(\tangentAngle)}
			\pgfmathsetmacro{\dy}{0.25 * sin(\tangentAngle)}
			\node (U\i) [rectangle, minimum height=0.1mm, minimum width=0.1mm, inner sep=0pt, rotate=\tangentAngle] at (\x-\dx,\y-\dy) {};
			\node (S\i) [rectangle, minimum height=0.1mm, minimum width=0.1mm, inner sep=0pt, rotate=\tangentAngle] at (\x+\dx,\y+\dy) {};
		}
		\pgfmathsetmacro{\angle}{360/12 * 2}
		\pgfmathsetmacro{\x}{\centerX + \radius * cos(\angle)}
		\pgfmathsetmacro{\y}{\centerY + \radius * sin(\angle)}
		\node (vt) [circle, draw, inner sep=0pt, minimum size=4mm] at (\x, \y) {$v_\mathrm{t}$};
		\foreach \x in {7, ...,10}{
			\draw (S\x) to (R\the\numexpr \x+1 \relax);
			\draw (S\x) to [bend left=40] (R\the\numexpr \x+1 \relax);
		}
		\foreach \x in {5,4}{
			\draw (U\x) to (R\the\numexpr \x-1 \relax);
			\draw (U\x) to [bend right=40] (R\the\numexpr \x-1 \relax);
		}
		
		\draw  (v2) to (R5);
		\draw  (v2) to [bend right=40] (R5);
		\draw  (v2) to (R7);
		\draw  (v2) to [bend left=40] (R7);
		\draw  (S11) to (R0);
		\draw  (S11) to [bend left=40] (R0);
		\draw  (S0) to (R1);
		\draw  (S0) to [bend left=40] (R1);
		\draw  (S1) to (vt);
		\draw  (S1) to [bend left=40] (vt);
		\draw  (U3) to (vt);
		\draw  (U3) to [bend right=40] (vt);
		
		\node at (2.5, -3.1) {$P_1$};
		\node at (8.1, -3.1) {$P_2$};
		\node at (2.5, 0) {$C_1$};
		\node at (8.1, 0) {$C_2$};
		\node at (13.7, 0) {$C^*$};
		\node at (0.7, 2.5) {$P_1'$};
		\node at (4.5, 2.5) {$P_1''$};
		\node at (6, 2.5) {$P_2'$};
		\node at (10, 2.5) {$P_2''$};
		
	\end{tikzpicture}
	\caption{The block arrangement from \cite{birx} used here in the proof of \cref{thm:4-lower-bound}. Rectangle shapes represent blocks and circle shapes represent single vertices. All drawn edges have weight $k^d$, where we have omitted edges inside the blocks.
		Pairs of edges beginning at the same point inside a block are its exit edges and a pair of edges leading to two different points in a block are its entry edges.
	}\label{fig:lower-bound-4}
\end{figure}

\begin{proof}[Proof of \cref{thm:4-lower-bound}]
	As a first step, we give the construction for the block arrangement as in \cite{birx}.\vspace{2mm}\\
	\noindent
	\textbf{Construction.}
	Let $\epsilon>0$ be arbitrary and choose  $d,k$ large enough (depending on~$\epsilon$) with $k\in\Omega(d^2)$.
	Fix an online graph exploration algorithm \alg.
	Whenever  we refer to a \emph{block} in this proof, we mean a block constructed by $\mathcal{A}_d$ as in \cref{lem:recursive block}.
	Moreover, let~$M$ and $N$ be a large enough integers.
	
	We construct a graph $G$ (adaptively depending on the behavior of $\alg$) as follows (see \cref{fig:lower-bound-4}).
	Let $v_\mathrm{s}$ denote the starting position and identify it with the starting vertices of two blocks, i.e., $v_\mathrm{s}$ is incident to four boundary edges of weight~$k^d$ and the first two such traversed boundary edges are the algorithmic entry edges of two different blocks.
	Whenever the agent successfully explores and leaves a block, its exit vertex is identified with the starting vertex of a new block, so that we obtain two chains of blocks $P_1$ and $P_1'$ and, in each block, we employ strategy $\mathcal{A}_d$.
	This procedure ends when one of the chains consists of $M$ blocks constructed by $\mathcal{A}_d$ and we assume w.l.o.g.~that this occurs for chain $P_1$.
	We let the exit edges of the last explored block lead to a vertex $v_\mathrm{1}$. 
	This vertex is identified with the starting vertex of three further blocks, i.e., $v_1$ is incident to 8 edges of weight~$k^d$ (two exit edges and 6 entry edges).
	The agent starts exploring these and, similarly as before, we build chains of blocks.
	Let $M'$ denote the number of blocks in $P_1'$, i.e., $P_1'$ consists of~$M'-1$ explored blocks and one partially explored block.
	Let $M''$ be the maximum number of blocks contained in any of the three chains starting from $v_1$.
	We stop the process when $M'+M''=M$ and the agent traverses a new exit edge.
	Let $P_1''$ be the chain starting from $v_1$ of length $M''$.
	When the process stops, there is exactly one block that was entered, but not explored: If the last block that was traversed is in $P_1'$, this partially explored block is in $P_1''$, and if the last block that was explored is in~$P_1''$, the partially explored block is in $P_1'$.
	We complete it to some valid block that $\mathcal{A}_d$ constructs (for some algorithm) and call this block $B_1^*$.
	We let the edges of the last blocks in $P_1'$ and $P_1''$ lead to the same vertex $u_\mathrm{1}$.
	
	We obtain a cycle of blocks consisting of $P_1, P_1', P_1''$, which we refer to as~$C_1$ and, from vertex~$v_1$, there are two chains of blocks with unexplored ends of length less than~$M$. We iterate the procedure by ignoring the explored cycle~$C_1$ and interpreting~$v_1$ as the new starting vertex $v_s$.
	This is repeated until we obtain~$N$ such cycles of blocks.
	Then, we have two chains of blocks with unexplored ends beginning from~$v_N$.
	As soon as their total numbers of blocks sum up to $2M$, we let the exit edges of the last blocks in these chains be connected to a closing vertex $v_\mathrm{t}$.
	We call this last cycle of blocks between $v_N$ and $v_\mathrm{t}$ cycle $C^*$.
	Observe that the construction is planar using that the blocks constructed by $\mathcal{A}_d$ are planar.\vspace{1mm}\\
	
	\noindent
	\textbf{Analysis.}
	First, note that every block, except for the blocks $B_i^*$ for each $i$, was only entered via an entry edge before it was explored.
	For each of these blocks~$B$, the agent incurs a cost of $\alg_B$ during the first exploration.
	Next, we will argue that almost every block will be traversed a second time. More precisely, we claim the following, where we use that all blocks have the same block length $2k^d$.
	\begin{claim}
		The agent incurs a total cost of at least $N(2M-1)\cdot 2k^d$ for backtracking blocks.
	\end{claim}
	\begin{proof}
		We only argue that the algorithm incurs a cost of $(2M-1)\cdot 2k^d$ for backtracking in $C_1$ and it will be clear that this holds similarly for all other cycles $C_2, \dots, C_N$ (but not for $C^*$).
		We distinguish two cases: $M'+M''=M$ was triggered by (1) exploration of a block in $P_1'$ or (2) by exploration of a block in $P_1''$.
		
		In case (1), observe that the agent already backtracked all $M$ blocks in $P_1$ at the time when~$M'+M''=M$ was triggered, so that it incurred a cost of at least $M\cdot 2k^d$.
		Then, the agent has to travel to~$v_1$ to reach cycle $C_2$.
		Since it has to return to the starting position when all vertices were explored, it needs to traverse a path from~$v_1$ to $v_s$ at the end.
		A shortest such path is to again traverse all blocks in $P_1$ so that the agent incurs another cost of $M\cdot 2k^d$ for backtracking in $C_1$.

		In case (2), observe that the agent backtracked all blocks in $P_1'\setminus \{B_1^*\}$ at the time when~$M'+M''=M$ was triggered. To proceed with the exploration, the agent has to travel to $v_1$ and the shortest way for this is by backtracking all blocks in chain~$P_1''$, so we can charge this cost by assuming that these blocks were backtracked.
		At the end, the agent has to return to the starting position, for which it needs to travel from $v_1$ to $v_\mathrm{s}$ and the shortest way for this is by backtracking all blocks in chain $P_1$.
		Together, we obtain that, in either case, we can assume that the agent backtracks all blocks in~$C_1\setminus \{B_1^*\}$ and we have $|C_1\setminus \{B_1^*\}|=(2M-1)$.
		
	\end{proof}
	
	The claim implies that, for every block, except for $B_i^*$ ($i\in\{1, \dots, N\}$) and the blocks in $C^*$, the agent incurs a cost of $\alg_B$ for the first traversal, and a cost of~$2k^d=l_B$ for backtracking.
	We obtain
	\begin{align}
		\alg(G, v_s)&\geq\sum\limits_{j=1}^N\sum\limits_{B \in C_j\setminus \{B_j^*\}}\alg_B+l_B
		=\sum\limits_{j=1}^N\sum\limits_{B \in C_j\setminus \{B_j^*\}}\core_B+\opt_B\nonumber\\
		&\geq (4-\epsilon)\cdot \sum\limits_{j=1}^N\sum\limits_{B \in C_j\setminus \{B_j^*\}}\opt_B,\label{eq:alg-lower-bound-final}
	\end{align}
	where we have used \cref{lem:recursive block} and that $k$ and $d$ are large enough (depending on~$\epsilon$).
	
	At the same time, in the offline optimum, the agent only needs to explore every block once and never backtracks a block.
	We obtain
	\begin{align}
		\opt(G)&\leq\sum\limits_{j=1}^{N}\left(3k^d+\sum\limits_{B \in C_j}\opt_B\right) + \sum\limits_{B \in C^*}\opt_B + 2k^d\nonumber \\
		&\leq \sum\limits_{j=1}^{N} \left(\sum\limits_{B \in C_j\setminus\{B_j^*\}}\opt_B \right)+ (2N+2M+1)(d+2)k^{d}\label{eq:opt-lower-bound-final},
	\end{align}
	where, for the last inequality, we have estimated each of $\opt_B, 3k^d, 2k^d$ by $(d+2)k^d$ (for $B=B_j^*$ and $B\in C^*$) and used that $C^*$ consists of $2M$ blocks.
	
	Note that 
	\begin{align*}
		&\frac{(2N+2M+1)(d+2)k^d}{\alg(G, v_s)}
		\leq \frac{(2N+2M+1)(d+2)k^d}{N\cdot (2M-1)2k^d} \to 0 \hspace{2mm}(M,N \to \infty),
	\end{align*}
	where we have used the safe lower bound $\alg(G, v_s)\geq N(2M-1)l_B=N(2M-1)2k^d$.

	By \eqref{eq:alg-lower-bound-final} and \eqref{eq:opt-lower-bound-final}, we obtain together with the last inequality for the competitive ratio that $\liminf_{M,N\to\infty}\alg(G, v_s)/\opt(G)\geq (4-\epsilon)$.
	Since $\epsilon>0$ was chosen arbitrarily, we obtain that the competitive ratio is at least 4 on planar graphs.
	Combining with the fact that any planar construction can be made subcubic (\cref{cor:exploration-restrictions}), this completes the proof of \cref{thm:4-lower-bound}.
\end{proof}

Last, we remark that our analysis for the constructed graph is tight.
To see this, one can easily verify that depth-first search (ignoring edge weights) is 4-competitive on our construction.

\paragraph{Acknowledgements.}
I thank Yann Disser for valuable feedback on the writeup of this paper.

 \bibliographystyle{alpha}
 \bibliography{References.bib}

\end{document}